%% file: main.tex
\RequirePackage{amsmath}
\documentclass[a4paper, 10pt, reqno]{amsart}
\usepackage{amssymb, url, color, pb-diagram, graphicx, amscd, pb-diagram,mathrsfs}
\usepackage[colorlinks=true, bookmarks=true, pdfstartview=FitH, pagebackref=true]{hyperref}
\usepackage[nodayofweek]{datetime}
\usepackage{comment}
\usepackage{enumerate,stmaryrd} 
\usepackage{euscript,enumitem}
\usepackage{graphicx}
\usepackage{mathabx}
\usepackage{leftindex}
\usepackage[space, compress, sort]{cite}
\usepackage{empheq}
\usepackage[usenames,dvipsnames]{xcolor}
\usepackage{multirow}
\usepackage{microtype}
\usepackage[colorlinks=true, bookmarks=true, pdfstartview=FitH, pagebackref=true]{hyperref}

\input{def}

\date{\today}

\keywords{Einstein constraint equations, non-constant mean curvature, conformal method, Lichnerowicz equation, compact manifold, prescribed scalar curvature}

\subjclass[2000]{53C21 (Primary), 35Q75, 53C80, 83C05 (Secondary)}

\begin{document}


\author{Armand Coudray \and Romain Gicquaud}
\address[A. Coudray and R. Gicquaud]{Institut Denis Poisson \\ UFR Sciences et Technologie \\
    Facult\'e de Tours \\ Parc de Grandmont\\ 37200 Tours \\ FRANCE}
\email{\href{mailto: A. Coudray <Armand.Coudray@univ-tours.fr>}{armand.coudray@univ-tours.fr}, \href{mailto: R. Gicquaud <Romain.Gicquaud@idpoisson.fr>}{romain.gicquaud@idpoisson.fr}}

\title[Uniqueness for the conformal method]{A new approach towards the construction of initial data in general relativity with positive Yamabe invariant and arbitrary mean curvature}
\keywords{Einstein constraint equations, non-CMC, conformal method, Banach fixed point theorem, positive Yamabe invariant, small TT-tensor}

\begin{abstract}
    This paper revisits the classical construction of initial data using the
    conformal method, as originally proposed by Holst, Nagy, and Tsogtgerel and
    later refined by Maxwell. We demonstrate that the existence of the solution can
    be proven using the Banach fixed point theorem, whereas the original proof
    relied on the Schauder fixed point theorem. This new approach has two main
    advantages: it guarantees the uniqueness of the solution to the equations of
    the conformal method as soon as one imposes a bound on the physical volume of
    it and it provides an explicit construction of the solution.
\end{abstract}

\date{\today}
\maketitle

\tableofcontents

\section{Introduction}\label{secIntro}
\input{intro}

\section{Estimates for the Lichnerowicz equation}\label{secLichnerowicz}
\input{estimates}

\section{Estimates for the vector equation}\label{secVector}
\input{vector}

\section{Estimates for the solution under a volume bound}\label{secVolumeBound}
\input{VolumeBound}

\section{Existence and uniqueness of the solution}\label{secUniqueness}
\input{uniqueness}


\bibliographystyle{amsplain}
\bibliography{biblio}

\end{document}

%% file: def.tex
\theoremstyle{plain}

\newtheorem{thm}{Theorem}
\newtheorem{theorem}[thm]{Theorem}

\newtheorem{lemma}[thm]{Lemma}

\newtheorem{proposition}[thm]{Proposition}

\theoremstyle{remark}

\theoremstyle{definition}

\newcounter{mnotecount}[section]

\renewcommand{\phi}{\varphi}
\renewcommand{\epsilon}{\varepsilon}

\newcommand{\bR}{\mathbb{R}}
\newcommand{\bL}{\mathbb{L}}

\newcommand{\cY}{\mathcal{Y}}

\newcommand{\ghat}{{\widehat{g}}}

\newcommand{\Khat}{\widehat{K}}

\newcommand{\gbar}{\overline{g}}

\renewcommand{\hbar}{\overline{h}}

\newcommand{\Abar}{\overline{A}}
\newcommand{\Bbar}{\overline{B}}

\newcommand{\phibar}{\overline{\varphi}}

\newcommand{\ubar}{\overline{u}}

\newcommand{\Sring}{\mathring{S}}

\newcommand{\definedas}{\mathrel{\raise.095ex\hbox{\rm :}\mkern-5.2mu=}}

\let\<\langle
\let\>\rangle

\DeclareMathOperator{\tr}{tr}
\DeclareMathOperator{\divg}{div}

\DeclareMathOperator{\vol}{Vol}

\newcommand{\scal}{\mathrm{Scal}}

\def\XXint#1#2#3{{\setbox0=\hbox{$#1{#2#3}{\int}$}
            \vcenter{\hbox{$#2#3$}}\kern-.5\wd0}}

%% file: intro.tex
The resolution of the Cauchy problem in general relativity by Y.~Choquet-Bruhat
and R.~Geroch~\cite{CB1, CB2} marked a decisive starting point for the
systematic construction of increasingly large classes of initial data sets. It
is by now well understood that Einstein’s equations are not hyperbolic \emph{a
    priori}: when formulated on a spacetime foliated by spacelike hypersurfaces,
part of the equations imposes restrictions on the admissible initial data
rather than governing their evolution.

These restrictions take the form of a coupled system of nonlinear elliptic
equations on the initial hypersurface, known as the \emph{constraint
    equations}. In the vacuum case, they read
\[
    \begin{cases}
        \scal^{\ghat} + \left(\tr_{\ghat} \Khat\right)^2 - \left|\Khat\right|_{\ghat}^2 & = 0, \\
        \divg_{\ghat} \Khat - d \left(\tr_{\ghat} \Khat\right)                          & = 0.
    \end{cases}
\]
Here, $\ghat$ is a Riemannian metric on the spacelike hypersurface $M$ and
$\Khat$ denotes its second fundamental form. We refer the reader
to~\cite{BartnikIsenberg} or~\cite{ChoquetBruhat} for introductory accounts of
this topic.

Among the various approaches that have been developed to solve the constraint
equations, the \emph{conformal method}, introduced by J.~York
in~\cite{YorkDecomposition}, has emerged as the most prominent and widely used
framework. We refer to~\cite{CarlottoReview} for a comprehensive overview of
the different methods that have been proposed.

The basic principle of the conformal method consists in prescribing part of the
initial data within a given conformal class and reducing the constraint
equations to a coupled elliptic system for a scalar conformal factor and a
vector field. More precisely, one sets
\[
    \begin{cases}
        \ghat & = \phi^\kappa g,                                      \\
        \Khat & = \dfrac{\tau}{n} \ghat + \phi^{-2} (\sigma + \bL W),
    \end{cases}
\]
where $g$ is a background Riemannian metric, $\kappa = \frac{4}{n-2}$, $\tau$
is the prescribed mean curvature, $\sigma$ is a transverse-traceless tensor
with respect to $g$, and $\bL$ denotes the conformal Killing operator
\[
    (\bL W)_{ij} = \nabla_i W_j + \nabla_j W_i - \frac{2}{n} \nabla^k W_k g_{ij}.
\]
The constraint equations are then equivalent to the following system:
\begin{subequations}\label{eqConstraints}
    \begin{align}
        - \frac{4(n-1)}{n-2} \Delta \phi + \scal\,\phi + \frac{n-1}{n} \tau^2 \phi^{N-1}
         & = \frac{|\sigma + \bL W|^2}{\phi^{N+1}}, \label{eqLichnerowicz} \\
        \Delta_{\bL} W
         & = \frac{n-1}{n} \phi^N d\tau, \label{eqVector}
    \end{align}
\end{subequations}
where $n$ denotes the dimension of the manifold $M$, $\Delta_{\mathbb L} :=
    -\frac 12\mathbb L^{\ast}(\mathbb L \cdot)$ and $N := \frac{2n}{n-2}$.

This system is commonly referred to as the \emph{conformal constraint
    equations}. Equation~\eqref{eqLichnerowicz} is known as the \emph{Lichnerowicz
    equation}, while~\eqref{eqVector} is usually called the \emph{vector equation}.

A first major breakthrough in the study of this system was achieved by
J.~Isenberg in 1995, who classified the solutions of the conformal constraint
equations in the constant mean curvature (CMC) setting in~\cite{Isenberg}. This
case is not only physically relevant, but also leads to a striking
simplification of the system. Indeed, when $\tau$ is constant, the vector
equation reduces to $\Delta_{\bL} W = 0$, which implies that $W \equiv 0$
provided that the metric $g$ admits no conformal Killing vector fields.

In this situation, one is thus left with the Lichnerowicz
equation~\eqref{eqLichnerowicz} alone, a problem that is by now well understood
on compact manifolds; see for instance~\cite{GicquaudLichnerowicz}.

Using the implicit function theorem, the solvability of the conformal
constraint equations in the near-CMC regime was established under various
assumptions by several authors, but the most relevant for us is the work by
P.~Allen, A.~Clausen and J.~Isenberg in~\cite{ACI08}.

In contrast, a strikingly different approach was introduced in 2008 by
M.~Holst, G.~Nagy and G.~Tsogtgerel~\cite{HNT1, HNT2}, and subsequently refined
by D.~Maxwell~\cite{MaxwellNonCMC}. While the mean curvature $\tau$ had long
been regarded as the main obstruction in the non-CMC setting, this method
showed that $\tau$ can in fact be chosen arbitrarily, provided the Yamabe
invariant of $g$ is positive and the transverse-traceless tensor $\sigma$ is
sufficiently small in $L^\infty$, but not identically zero. The regularity
assumptions on $\sigma$ were later weakened by T.~C.~Nguyen~\cite{Nguyen}.

This approach was subsequently reinterpreted by the second author and A.~Ngô
in~\cite{GicquaudNgo} as a perturbative argument near $\tau = 0$, followed by a
rescaling procedure, thereby paving the way for further generalizations, see
e.g.~\cite{GicquaudNguyen,GicquaudSmallTT}. One of the main questions left open
by this method concerned the uniqueness of solutions. As shown
in~\cite{Nguyen2}, global uniqueness cannot be expected in general. However,
the second author proved in~\cite{GicquaudUniqueness} that uniqueness does
hold, under the technical assumption that $|\sigma|$ is bounded away from zero,
provided the physical volume
\[
    V \definedas \vol(M, \ghat) = \int_M \phi^{N} \, d\mu^g
\]
does not exceed an explicit threshold.

The purpose of the present paper is to show that the Schauder fixed point
theorem at the heart of the approach introduced in~\cite{HNT1, HNT2,
    MaxwellNonCMC} can be replaced by the Banach contraction mapping theorem. This
is a significant strengthening: whereas the Schauder theorem is
non-constructive and yields existence alone, the Banach theorem simultaneously
establishes existence and uniqueness, and provides a convergent iterative
scheme to approximate the solution. As a byproduct, this approach also allows
us to remove the technical assumption of~\cite{GicquaudUniqueness} that
$|\sigma|$ is bounded away from zero.

Let $p > \frac{n}{2}$ be given. In what follows, we will make the following
regularity assumptions on the seed data $(M, g)$, $\sigma$ and $\tau$:
\begin{itemize}
    \item The metric $g$ belongs to $W^{2, p}(M, S_2M)$, has positive Yamabe invariant,
          and admits no non-trivial conformal Killing vector fields, i.e.\ for any vector
          field $V \in W^{1, 2}(M, TM)$, $\bL V \equiv 0 \Rightarrow V \equiv 0$.
    \item The mean curvature $\tau$ belongs to $L^\infty(M, \bR) \cap W^{1, n}(M, \bR)$,
    \item The TT-tensor $\sigma$ belongs to $L^{2p}(M, \Sring_2 M)$.
\end{itemize}

We remind the reader that the Yamabe invariant $\cY(M, g)$ of the metric $g$ is
defined as follows:
\[
    \cY(M, g) \definedas \inf_{\substack{u \in W^{1, 2}(M, \bR)\\ u \not\equiv 0}} \frac{\int_M \left[\frac{4(n-1)}{n-2} |du|^2 + \scal\, u^2\right] d\mu^g}{\left(\int_M |u|^N d\mu^g\right)^{2/N}}.
\]
In particular, the weak form of the Yamabe theorem states that $(M, g)$ has
positive Yamabe invariant if and only if there exists a metric $\gbar =
    \psi^\kappa g$ in the conformal class of $g$ having scalar curvature bounded
from below by a positive constant; see~\cite{LeeParker} for further details.

The result we prove in this paper is the following:
\begin{theorem}\label{thmMain}
    Let $(M, g)$ be a compact Riemannian manifold, $\tau$ a given function on $M$
    and $\sigma$ a non-zero TT-tensor for $g$ satisfying the regularity assumptions
    stated above. Let $V_{\max}$ be a small enough positive constant and $\omega_0
        > 0$. Then there exists a constant $c = c(M, g, \tau, \omega_0, V_{\max}, p) >
        0$ such that, if
    \[
        \|\sigma\|_{L^{2p}} \leq c \quad\text{and}\quad \|\sigma\|_{L^{2p}} \leq \omega_0 \|\sigma\|_{L^2},
    \]
    there exists a unique solution $(\phi, W) \in W^{2, p}(M, \bR) \times W^{2,
                p}(M, TM)$ to the system~\eqref{eqConstraints} such that
    \[
        V(\phi, W) \definedas \int_M \phi^N d\mu^g \leq V_{\max}.
    \]
\end{theorem}

The outline of the paper is as follows. In Section~\ref{secLichnerowicz}, we
establish estimates for the Lichnerowicz equation~\eqref{eqLichnerowicz},
including a lower bound on its solution. Section~\ref{secVector} is devoted to
estimates for the vector equation~\eqref{eqVector}. In
Section~\ref{secVolumeBound}, we show that, under a volume bound, the norm of
any solution $(\phi, W)$ to the conformal constraint equations is controlled by
the norm of $\sigma$. Finally, in Section~\ref{secUniqueness}, we prove that,
for $\sigma$ small enough, the fixed point map $\Phi$ is a contraction on a
suitable complete metric space, and deduce the existence and uniqueness of the
solution from the Banach fixed point theorem.

\subsection*{Acknowledgments}

This work was partially supported by the French National Research Agency (ANR)
under grants ANR-23-CE40-0010-02 (Einstein constraints: past, present, and
future, EINSTEIN-PPF) and ANR-25-CE40-4883 (Scattering, Holography and General
Relativity). Armand Coudray is grateful to the Institut Denis Poisson for its
hospitality.

%% file: estimates.tex
The goal of this section is to obtain estimates for the solution $\phi$ to the
Lichnerowicz equation~\eqref{eqLichnerowicz} in the case where the manifold
$(M, g)$ has a positive Yamabe invariant. To slightly lighten notations, we set
$A \definedas |\sigma + \bL W|$. So we study the following elliptic equation:

\begin{equation} \label{eqLich}
    - \frac{4(n-1)}{n-2} \Delta \phi + \scal\,\phi + \frac{n-1}{n} \tau^2 \phi^{N-1} = \frac{A^2}{\phi^{N+1}}.
\end{equation}

We assume, in what follows, that $A \in L^{2p}(M, \bR)$, $A \not\equiv 0$.
Existence and uniqueness of the solution to the Lichnerowicz equation is by now
standard, see e.g.~\cite{GicquaudLichnerowicz}, as well as the continuity of
the mapping sending $A \in L^{2p}$ to $\phi \in W^{2, p}(M, \bR)$, see
e.g.~\cite{DahlGicquaudHumbert}. Our first goal is to obtain $L^p$-estimates
for positive powers of $\phi$. These estimates are well established now (see
e.g.\ \cite[Lemma 2.11 and Proposition 2.12]{Pailleron}) but we give a proof of
it for the sake of completeness.

\begin{proposition}\label{propEstimateLichUpper}
    Let $\phi$ be the solution to Equation~\eqref{eqLich}, with $A \in L^q(M,
        \bR)$, $q \geq 2$, $A \not\equiv 0$. Then,

    \begin{itemize}
        \item If $q < n$, we have $\left\|\phi\right\|_{L^r}^{N+2} \lesssim \|A\|_{L^q}^2$
              with $r$ such that
              \[
                  \frac{1}{q} = \frac{1}{n} + \frac{2(n-1)}{n-2} \frac{1}{r}.
              \]
        \item If $q \geq n$, we have $\left\|\phi\right\|_{L^r}^{N+2} \lesssim \|A\|_{L^q}^2$
              for any $r \geq N$.
    \end{itemize}
\end{proposition}

In particular, the proposition shows that there exists a constant $\mu_L > 0$
such that, for any $A \in L^2(M, \bR)$, $A \not\equiv 0$, the solution $\phi$
to the Lichnerowicz equation~\eqref{eqLich} satisfies
\begin{equation}\label{eqDefMuL}
    \mu_L \|\phi^N\|_{L^{\frac{N}{2}+1}}^{2\frac{n-1}{n}} \leq \|A\|_{L^2}^2.
\end{equation}

\begin{proof}
    From~\cite{LeeParker}, we know that there exists a positive function $\psi \in
        W^{2, p}(M, \bR)$ such that the metric $\gbar \definedas \psi^\kappa g$ has
    scalar curvature bounded from below by a positive constant $\mu$:
    \[
        \scal_{\gbar} \geq \mu > 0 \quad \text{a.e.}.
    \]

    For any function $u \in W^{2, p}(M, \bR)$, the conformal transformation law of
    the conformal Laplacian reads:
    \begin{equation}\label{eqConfTrans}
        - \frac{4(n-1)}{n-2} \Delta_{\gbar} \ubar + \scal_{\gbar}\,\ubar = \psi^{-1-\kappa} \left(- \frac{4(n-1)}{n-2} \Delta u + \scal\,u\right),
    \end{equation}
    where $\ubar \definedas \psi^{-1} u$. We rewrite Equation~\eqref{eqLich} with
    respect to the metric $\gbar = \psi^\kappa g$, setting $\phibar =
        \psi^{-1}\phi$ and $\Abar = \psi^{-N} A$. Using the conformal transformation
    law~\eqref{eqConfTrans}, we obtain
    \[
        -\frac{4(n-1)}{n-2}\Delta_{\gbar}\phibar
        + \scal_{\gbar}\,\phibar
        + \frac{n-1}{n}\tau^2 \phibar^{N-1}
        = \frac{\Abar^2}{\phibar^{N+1}}.
    \]

    Let $\alpha \geq \frac{N}{2} + 1$. Multiplying the equation by
    $\phibar^{2\alpha-1}$ and integrating over $M$, an integration by parts yields
    \[
        \int_M \Abar^2 \phibar^{2\alpha-N-2}\, d\mu^{\gbar}
        \ge
        c_\alpha \int_M |d\phibar^\alpha|_{\gbar}^2\, d\mu^{\gbar}
        + \int_M \scal_{\gbar}\, \phibar^{2\alpha}\, d\mu^{\gbar},
    \]
    where
    \[
        c_\alpha
        = \frac{4(n-1)}{n-2}\frac{2\alpha-1}{\alpha^2}.
    \]
    Since $\scal_{\gbar} \ge \mu > 0$, the right-hand side controls the $W^{1,
                2}$-norm of $\phibar^\alpha$.

    We now estimate the left-hand side using H\"older's inequality. Let $a,b>1$
    satisfy $\frac{1}{a} + \frac{1}{b} = 1$, and choose $b$ such that
    \[
        (2\alpha - N - 2)b = N\alpha.
    \]
    Then
    \[
        \int_M \Abar^2 \phibar^{2\alpha-N-2}\, d\mu^{\gbar}
        \le
        \|\Abar\|_{L^{2a}(M,\gbar)}^2
        \left(\int_M \phibar^{N\alpha}\, d\mu^{\gbar}\right)^{1/b}.
    \]
    A straightforward computation gives
    \[
        \frac{1}{b}
        = \frac{2\alpha - N - 2}{N\alpha}
        = \frac{n-2}{n} - \frac{2(n-1)}{n\alpha},
        \qquad
        \frac{1}{2a}
        = \frac{1}{n} + \frac{n-1}{n\alpha}.
    \]
    By the Sobolev inequality on $(M,\gbar)$, there exists a constant $C>0$ such
    that
    \[
        \left(\int_M \phibar^{N\alpha}\, d\mu^{\gbar}\right)^{2/N}
        \le
        C \left(
        \int_M |d\phibar^\alpha|_{\gbar}^2\, d\mu^{\gbar}
        + \int_M \phibar^{2\alpha}\, d\mu^{\gbar}
        \right).
    \]
    Combining the previous inequalities yields
    \[
        \left(\int_M \phibar^{N\alpha}\, d\mu^{\gbar}\right)^{\frac{2(n-1)}{n\alpha}}
        \lesssim
        \|\Abar\|_{L^{2a}(M,\gbar)}^2.
    \]
    Since $\psi$ is uniformly bounded from above and below on $M$, we conclude that
    \[
        \|\phi\|_{L^{N\alpha}(M,g)}^{N+2}
        \lesssim
        \|A\|_{L^q(M,g)}^2,
        \qquad q = 2a < n.
    \]
    This proves the first case of the proposition.

    If $q \ge n$, then $A \in L^{q'}(M, \bR)$ for every $q' < n$, and the previous
    argument applies with any such $q'$, yielding the desired estimate for all $r
        \ge N$.
\end{proof}

Proving that the solution $\phi$ to the Lichnerowicz equation~\eqref{eqLich} is
bounded away from zero by some explicit constant requires some work. The main
ingredient we use was found by Maxwell in~\cite{MaxwellNonCMC}:

\begin{lemma}\label{lmGreenLgtau}
    Let $g\in W^{2,p}(M, S_2M)$ with $p > \frac{n}{2}$ on a compact manifold $M^n$ ($n\ge3$), and let
    $\tau\in L^\infty(M, \bR)$. For $q\in(1,p]$, consider the operator
    \[
        L_{g,\tau} : W^{2,q}(M, \bR)\to L^q(M, \bR),\qquad
        L_{g,\tau}(v):= -\frac{4(n-1)}{n-2} \Delta_g v + \scal_g\, v + \frac{n-1}{n}\tau^2 v,
    \]
    Then $L_{g,\tau}$ is an isomorphism for every $q\leq p$ whose inverse is given
    by a Green kernel $G_\tau(x,y)$: for any $f\in L^q(M, \bR)$ the unique solution
    $v\in W^{2,q}(M, \bR)$ of $L_{g,\tau}v=f$ satisfies
    \[
        v(x)=\int_M G_\tau(x,y)f(y)\,d\mu^g(y)\qquad \text{for a.e.\ }x\in M.
    \]
    Moreover, if $q>\frac{n}{2}$, the identity holds for all $x\in M$. Finally,
    there exists a constant $m_{g, \tau} > 0$ such that
    \[
        G_\tau(x,y)\geq m_{g,\tau} \quad \forall x, y \in M, x \neq y.
    \]
\end{lemma}

\begin{proof}
    The existence of a Green function $G_\tau$ for operators similar to $L_{g,
                \tau}$ is addressed in~\cite{KimSakellaris}, see also~\cite{AvalosCogoAbrego}.
    Note that, as $L_{g, \tau}$ is formally selfadjoint, $G_\tau$ is symmetric and,
    for each $x_0 \in M$, we have $L_{g, \tau} G_\tau(x_0, \cdot) = 0$ away from
    $x_0$. By elliptic regularity, for any $y_0 \in M \setminus \{x_0\}$ and any $r
        > 0$ such that $x_0 \not\in \Bbar_r(y_0)$, we have $G_\tau(x_0, \cdot) \in
        W^{2, p}(\Bbar_r(y_0), \bR)$.

    Let $K_1, K_2 \subset M$ be compact and disjoint. For each $x \in K_1$,
    $G_\tau(x, \cdot)$ satisfies $L_{g, \tau} G_\tau(x, \cdot) = 0$ on a
    neighborhood of $K_2$. From~\cite[Theorem 6.12]{KimSakellaris}, $G_\tau(x, y)
        \lesssim d_g(x, y)^{2-n}$, which gives a uniform upper bound on $G_\tau$ over
    $K_1 \times K_2$ since $d(K_1, K_2) > 0$. Interior elliptic estimates then
    yield a uniform $W^{2, p}$-bound on $G_\tau(x, \cdot)|_{K_2}$ for all $x \in
        K_1$. Since $p > n/2$, the Sobolev embedding $W^{2, p} \hookrightarrow C^{0,
                \alpha}$ promotes this to a uniform $C^{0, \alpha}$-bound, so the family
    $\{G_\tau(x, \cdot)\}_{x \in K_1}$ is equicontinuous on $K_2$. By the symmetry
    $G_\tau(x, y) = G_\tau(y, x)$, the family $\{G_\tau(\cdot, y)\}_{y \in K_2}$ is
    likewise equicontinuous on $K_1$. The joint continuity of $G_\tau$ on $K_1
        \times K_2$ follows, and since $K_1, K_2$ are arbitrary, $G_\tau$ is continuous
    on $(M \times M) \setminus \Delta$, where $\Delta = \{(x, x) : x \in M\}$ is
    the diagonal. From the lower bound $G_\tau(x, y) \geq c_n d_g(x, y)^{2-n}$ for
    $x, y$ sufficiently close in $M$, see~\cite[Lemma 6.4]{Mourgoglou}, we see that
    $G_\tau$ is proper on $(M \times M)\setminus \Delta$. As a consequence, there
    exists a point $(x_0, y_0) \in M \times M$ where $G_\tau$ reaches its minimum
    value $m_{g, \tau} = G_\tau(x_0, y_0)$.

    As the Green function $G_\tau$ is non-negative, we have $m_{g, \tau} \geq 0$.
    As $u: y \mapsto G_\tau(x_0, y)$ solves $L_{g, \tau} u = 0$ away from $x_0$, we
    see from the strong maximum principle (see~\cite[Theorem
        8.19]{GilbargTrudinger}) applied to $-u$ that $m_{g, \tau} > 0$.
\end{proof}

We can now construct a subsolution to the Lichnerowicz equation. Let $v \in
    W^{2, p}(M, \bR)$ be the unique solution to $L_{g, \tau} v = A^2$:
\begin{equation}\label{eqSubsolution}
    -\frac{4(n-1)}{n-2} \Delta v + \scal\,v + \frac{n-1}{n} \tau^2 v = A^2.
\end{equation}
From Lemma~\ref{lmGreenLgtau}, we have, for any $x \in M$,
\begin{equation}\label{eqEstimateInfV}
    v(x) = \int_M G_\tau(x, y) A^2(y) d\mu^g(y) \geq m_{g, \tau} \int_M A^2(y) d\mu^g(y) = m_{g, \tau} \|A\|_{L^2}^2.
\end{equation}

On the other hand, since $L_{g, \tau}: W^{2, p}(M, \bR) \to L^p(M, \bR)$ is
invertible and since the embedding $W^{2, p}(M, \bR) \hookrightarrow
    L^\infty(M, \bR)$ is continuous, there exists a constant $C_{g, \tau} > 0$
independent of $A \in L^{2p}(M, \bR)$ such that
\begin{equation}\label{eqEstimateSupV}
    \|v\|_{L^\infty} \leq C_{g, \tau} \|A\|_{L^{2p}}^2.
\end{equation}

We next construct a subsolution to the Lichnerowicz equation~\eqref{eqLich};
the maximum principle will then yield the lower bound for $\phi$, see
Proposition~\ref{propLowerBoundLich}. We note that the bound involves the ratio
$\omega$ of two Lebesgue norms of $A$, which will need to be controlled in the
proof of Theorem~\ref{thmMain}.

\begin{lemma}\label{lmSubsolution}
    Assume that $\|A\|_{L^{2p}} \le C_{g,\tau}^{-1/2}$. Then the function
    \[
        \phi_- \definedas \|v\|_{L^\infty}^{-\frac{N+1}{N+2}} v
    \]
    is a subsolution of~\eqref{eqLich}. Moreover, setting
    \[
        \omega \definedas \frac{\|A\|_{L^{2p}}}{\|A\|_{L^2}},
    \]
    there exists a constant $c_{g,\tau} > 0$ such that
    \[
        \phi_- \ge c_{g,\tau}\,
        \omega^{-\frac{3n-2}{2(n-1)}} \|A\|_{L^2}^{\frac{n-2}{2(n-1)}}.
    \]
\end{lemma}

\begin{proof}
    Set $\lambda \definedas \|v\|_{L^\infty}^{-\frac{N+1}{N+2}}$ and $\phi_- \definedas \lambda v$.
    Using~\eqref{eqSubsolution}, we have
    \[
        -\frac{4(n-1)}{n-2}\Delta v+\scal\,v = A^2 - \frac{n-1}{n}\tau^2 v.
    \]
    The subsolution inequality is therefore equivalent to
    \[
        (\lambda^{N+2}v^{N+1}-1)A^2
        +\frac{n-1}{n}\tau^2\bigl(\lambda^{2N}v^{2N}-\lambda^{N+2}v^{N+2}\bigr)\le 0.
    \]
    It is sufficient to require
    \[
        \lambda^{N+2}v^{N+1}\le 1
        \quad\text{and}\quad
        \lambda^{2N}v^{2N}\le \lambda^{N+2}v^{N+2}.
    \]
    The second condition is equivalent to $\lambda v\le 1$ since $N>2$.

    With the choice $\lambda=\|v\|_{L^\infty}^{-\frac{N+1}{N+2}}$, one has
    \[
        \lambda^{N+2}\|v\|_{L^\infty}^{N+1}=1,
    \]
    hence $\lambda^{N+2}v^{N+1}\le 1$ everywhere. Moreover,
    \[
        \lambda\|v\|_{L^\infty} = \|v\|_{L^\infty}^{\frac{1}{N+2}}\le 1
    \]
    provided $\|v\|_\infty\le 1$, which holds if $\|A\|_{L^{2p}} \le
        C_{g,\tau}^{-1/2}$ by~\eqref{eqEstimateSupV}. This proves that $\phi_-=\lambda
        v$ is a subsolution.

    Finally, combining~\eqref{eqEstimateInfV} and~\eqref{eqEstimateSupV} yields
    \[
        \phi_- = \lambda v
        \ge
        \frac{m_{g,\tau}\|A\|_{L^2}^2}{\bigl(C_{g,\tau}\|A\|_{L^{2p}}^2\bigr)^{\frac{N+1}{N+2}}}
        =
        \frac{m_{g,\tau}}{C_{g,\tau}^{\frac{N+1}{N+2}}}\,
        \omega^{-\frac{3n-2}{2(n-1)}} \|A\|_{L^2}^{\frac{n-2}{2(n-1)}},
    \]
    which is the claimed bound with $c_{g,\tau} \definedas m_{g,\tau}
        C_{g,\tau}^{-\frac{N+1}{N+2}}$.
\end{proof}

We can now apply the maximum principle to obtain the following proposition:

\begin{proposition}\label{propLowerBoundLich}
    Under the assumptions of Lemma~\ref{lmSubsolution}, there exists a constant
    $\mu_{g,\tau} > 0$ such that, for any $A \in L^{2p}(M, \bR)$, $A \not\equiv 0$,
    the solution $\phi$ to the Lichnerowicz equation~\eqref{eqLich} satisfies,
    with $\omega$ as in Lemma~\ref{lmSubsolution},
    \[
        \phi \geq \mu_{g,\tau}\,
        \omega^{-\frac{3n-2}{2(n-1)}} \|A\|_{L^2}^{\frac{n-2}{2(n-1)}}.
    \]
\end{proposition}

\begin{proof}
    In view of Lemma~\ref{lmSubsolution} it suffices to prove that the solution $\phi$ to the Lichnerowicz equation~\eqref{eqLich} satisfies $\phi \geq \phi_-$. This is done by means of the maximum principle. Let $\psi \definedas \log \phi$ (resp. $\psi_- \definedas \log \phi_-$). Then $\psi \in W^{2, p}(M, \bR)$ (resp. $\psi_- \in W^{2, p}(M, \bR)$) satisfies
    \begin{align*}
         & - \frac{4(n-1)}{n-2} \left(\Delta \psi + |d\psi|^2\right) + \scal + \frac{n-1}{n} \tau^2 e^{(N-2)\psi} = A^2 e^{-(N+2) \psi}                                        \\
         & \left(\text{resp.}\ - \frac{4(n-1)}{n-2} \left(\Delta \psi_- + |d\psi_-|^2\right) + \scal + \frac{n-1}{n} \tau^2 e^{(N-2)\psi_-} \leq A^2 e^{-(N+2) \psi_-}\right).
    \end{align*}
    Subtracting the equation for $\psi$ and the inequation for $\psi_-$, we get
    \begin{align*}
         & - \frac{4(n-1)}{n-2} \left(\Delta (\psi - \psi_-) + \< d(\psi+\psi_-), d(\psi-\psi_-)\>\right)                                              \\
         & \qquad\qquad + \frac{n-1}{n} \tau^2 \left(e^{(N-2)\psi} - e^{(N-2)\psi_-}\right) \geq A^2 \left(e^{-(N+2) \psi} - e^{-(N+2) \psi_-}\right).
    \end{align*}
    We write
    \begin{align*}
        e^{(N-2)\psi} - e^{(N-2)\psi_-}
         & = \int_{\psi_-}^\psi (N-2) e^{(N-2)t} dt                            \\
         & = (N-2) (\psi - \psi_-) \int_0^1 e^{(N-2)(s\psi + (1-s)\psi_-)} ds,
    \end{align*}
    where we made the change of variable $t = s \psi + (1-s) \psi_-$. And similarly,
    \begin{align*}
        e^{-(N+2)\psi} - e^{-(N+2)\psi_-}
         & = - \int_{\psi_-}^\psi (N+2) e^{-(N+2)t} dt                            \\
         & = - (N+2) (\psi - \psi_-) \int_0^1 e^{-(N+2)(s\psi + (1-s)\psi_-)} ds.
    \end{align*}
    All in all, the difference $\psi - \psi_-$ satisfies the following differential inequality:
    \[
        - \frac{4(n-1)}{n-2} \left(\Delta (\psi - \psi_-) + \< d(\psi+\psi_-), d(\psi-\psi_-)\>\right) + B^2 (\psi - \psi_-) \geq 0,
    \]
    with
    \begin{align*}
        B^2 & \definedas \frac{n-1}{n} \tau^2 (N-2) \int_0^1 e^{(N-2)(s\psi + (1-s)\psi_-)} ds \\
            & \qquad\qquad + (N+2) A^2 \int_0^1 e^{-(N+2)(s\psi + (1-s)\psi_-)} ds.
    \end{align*}
    From the maximum principle~\cite[Theorem 8.1]{GilbargTrudinger}, we conclude
    that $\phi \geq \phi_-$.
\end{proof}

%% file: vector.tex
In this short section, we collect basic facts about the vector
equation~\eqref{eqVector}. These facts have already appeared in many places in
the literature. We follow here the presentation of~\cite{GicquaudUniqueness}.
The following proposition is borrowed from~\cite[Proposition 3.1]{Pailleron}.

\begin{proposition}\label{propVector}
    Assume that $(M, g)$ has no non-trivial conformal Killing vector field, i.e.\
    no non-zero vector field $V$ such that $\bL V \equiv 0$. Then the operator
    $\Delta_{\bL}: W^{2, q}(M, TM) \to L^q(M, TM)$ is an isomorphism for all $q \in
        (1, p]$.
\end{proposition}

Of particular importance in the next section will be the following constant
$\mu_V$:

\begin{lemma}\label{lmVector}
    Assume that $(M, g)$ has no non-trivial conformal Killing vector field. The constant $\mu_V$ defined by
    \begin{equation}\label{eqDefMuV}
        \mu_V \definedas \inf_{\substack{V \in W^{1, 2}(M, TM),\\ V \not\equiv 0}}
        \frac{\frac{1}{2}\int_M |\bL V|^2 d\mu^g}{\left(\int_M |V|^N d\mu^g\right)^{2/N}},
    \end{equation}
    is strictly positive.
\end{lemma}

%% file: VolumeBound.tex
In~\cite{GicquaudUniqueness} it was shown that, given a maximal allowable
volume $V_{\max}$ below an explicit threshold, the size of a solution
$(\phi,W)$ is explicitly controlled by the size of the TT-tensor $\sigma$. This
can be understood as a gap phenomenon: there cannot exist a solution whose
volume is bounded by $V_{\max}$ while $\|\phi^N\|_{L^{\frac{N}{2}+1}}$ is too
large compared to $\|\sigma\|_{L^2}$. The next proposition is a quantitative
version of this fact.

\begin{proposition}\label{propVolumeBound}
	Let $\mu_L$ and $\mu_V$ be the constants defined in~\eqref{eqDefMuL}
	and~\eqref{eqDefMuV}. Let $V_{\max}>0$ be such that
	\[
		\frac{2}{\mu_V}\left(\frac{n-1}{n}\right)^2 \|d\tau\|_{L^n}^2\, V_{\max}^{2/n}
		< \mu_L.
	\]
	Let $(\phi,W)$ be a solution of the conformal constraint
	equations~\eqref{eqConstraints} such that $V(\phi,W)\leq V_{\max}$. Then
	\begin{equation}\label{eqVolumeBound}
		\left(
		\mu_L
		-\frac{2}{\mu_V}\left(\frac{n-1}{n}\right)^2 \|d\tau\|_{L^n}^2\, V_{\max}^{2/n}
		\right)
		\|\phi^N\|_{L^{\frac N2+1}}^{2\frac{n-1}{n}}
		\leq\|\sigma\|_{L^2}^2.
	\end{equation}
\end{proposition}

\begin{proof}
	We first estimate $\|\bL W\|_{L^2}^2$ in terms of $\|\phi^N\|_{L^{\frac N2+1}}$
	and $\|d\tau\|_{L^n}$. Multiply the vector equation~\eqref{eqVector} by $W$,
	integrate over $M$, and use the identity
	\[
		\int_M \langle \Delta_{\bL}W, W\rangle\,d\mu^g
		= \frac{1}{2} \int_M |\bL W|^2\,d\mu^g
	\]
	which follows by integration by parts. We obtain
	\begin{align}
		\frac{1}{2} \int_M |\bL W|^2\,d\mu^g
		 & = -\frac{n-1}{n}\int_M \phi^N \langle d\tau, W\rangle\,d\mu^g \nonumber \\
		 & \leq\frac{n-1}{n}\|\phi^N\|_{L^2}\,\|d\tau\|_{L^n}\,\|W\|_{L^N}.
		\label{eqLW2_estimate}
	\end{align}

	Next we interpolate the $L^2$-norm of $\phi^N$ between $L^1$ and $L^{\frac
				{N}{2}+1}$. Since $N=\frac{2n}{n-2}$ one checks that
	\[
		\frac{1}{2} = \left(1-\frac{1}{n}\right)\frac{1}{\frac{N}{2}+1} + \frac{1}{n}\cdot 1.
	\]
	Hence, by H\"older interpolation,
	\[
		\|\phi^N\|_{L^2}
		\le
		\|\phi^N\|_{L^1}^{1/n}\,
		\|\phi^N\|_{L^{\frac N2+1}}^{(n-1)/n}
		\le
		V_{\max}^{1/n}\,
		\|\phi^N\|_{L^{\frac N2+1}}^{(n-1)/n},
	\]
	where we used $\|\phi^N\|_{L^1} = V(\phi,W) \leq V_{\max}$.

	Insert this bound into~\eqref{eqLW2_estimate} to get
	\begin{equation}\label{eqLW2_estimate2}
		\frac12\int_M |\bL W|^2\,d\mu^g
		\le
		\frac{n-1}{n}V_{\max}^{1/n}\,
		\|\phi^N\|_{L^{\frac N2+1}}^{(n-1)/n}\,
		\|d\tau\|_{L^n}\,
		\|W\|_{L^N}.
	\end{equation}

	We now use the definition of $\mu_V$ in~\eqref{eqDefMuV}, which yields
	\[
		\|W\|_{L^N}^2
		\leq\frac{1}{2\mu_V}\int_M |\bL W|^2\,d\mu^g
		\qquad\Longrightarrow\qquad
		\|W\|_{L^N}
		\leq\left(\frac{1}{2\mu_V}\int_M |\bL W|^2\,d\mu^g\right)^{1/2}.
	\]
	Plugging this into~\eqref{eqLW2_estimate2} and rearranging gives
	\[
		\int_M |\bL W|^2\,d\mu^g
		\le
		\frac{2}{\mu_V}\left(\frac{n-1}{n}\right)^2
		V_{\max}^{2/n}\,
		\|\phi^N\|_{L^{\frac N2+1}}^{2\frac{n-1}{n}}\,
		\|d\tau\|_{L^n}^2.
	\]
	This is the desired bound on $\|\bL W\|_{L^2}^2$. We conclude by combining this
	bound with~\eqref{eqDefMuL}. Recall that $A=|\sigma+\bL W|$. Since $\sigma$ is
	$L^2$-orthogonal to $\bL W$, York's decomposition yields
	\[
		\|A\|_{L^2}^2
		=\int_M |\sigma+\bL W|^2\,d\mu^g
		=\|\sigma\|_{L^2}^2 + \|\bL W\|_{L^2}^2.
	\]
	Therefore
	\[
		\|A\|_{L^2}^2
		\le
		\|\sigma\|_{L^2}^2
		+
		\frac{2}{\mu_V}\left(\frac{n-1}{n}\right)^2
		V_{\max}^{2/n}\,
		\|\phi^N\|_{L^{\frac N2+1}}^{2\frac{n-1}{n}}\,
		\|d\tau\|_{L^n}^2.
	\]
	On the other hand, by~\eqref{eqDefMuL},
	\[
		\mu_L\|\phi^N\|_{L^{\frac N2+1}}^{2\frac{n-1}{n}}
		\leq\|A\|_{L^2}^2.
	\]
	Combining the last two inequalities yields
	\[
		\left(
		\mu_L
		-\frac{2}{\mu_V}\left(\frac{n-1}{n}\right)^2
		V_{\max}^{2/n}\,
		\|d\tau\|_{L^n}^2
		\right)
		\|\phi^N\|_{L^{\frac N2+1}}^{2\frac{n-1}{n}}
		\leq\|\sigma\|_{L^2}^2,
	\]
	which concludes the proof of the proposition.
\end{proof}

%% file: uniqueness.tex
Let $r \definedas 2p N$. We define a map $\Phi: L^r(M, \bR_+) \to L^r(M,
	\bR_+)$ as follows (where we denote by $L^r(M, \bR_+)$ the set of non-negative
functions $\phi \in L^r(M, \bR)$). Given $\phi \in L^r(M, \bR)$, we let $W =
	\mathrm{Vect}(\phi)$ denote the unique solution to the vector
equation~\eqref{eqVector}:
\[
	\Delta_{\bL} W = \frac{n-1}{n} \phi^N d\tau.
\]
By elliptic regularity (Proposition~\ref{propVector}), we have
\[
	\|W\|_{W^{2, q}} \lesssim \left\|\phi^N d\tau\right\|_{L^q} \leq \|\phi^N\|_{L^{2p}} \|d\tau\|_{L^n},
\]
with $q \in (1, \infty)$ given by
\[
	\frac{1}{q} = \frac{1}{2p} + \frac{1}{n} < \frac{2}{n}.
\]
From the Sobolev embedding theorem, we conclude that
\[
	\|\bL W\|_{L^{2p}} \lesssim \|\phi^N\|_{L^{2p}} \|d\tau\|_{L^n} = \|\phi\|_{L^r}^N \|d\tau\|_{L^n}.
\]
Next, given $W$ such that $\bL W \in L^{2p}(M, TM)$, we let $\phi' =
	\mathrm{Lich}(W)$ be the unique solution to the Lichnerowicz
equation~\eqref{eqLichnerowicz}:
\[
	-\frac{4(n-1)}{n-2} \Delta \phi' + \scal \phi' + \frac{n-1}{n} \tau^2 (\phi')^{N-1} = \frac{|\sigma + \bL W|^2}{(\phi')^{N+1}}.
\]
From Proposition~\ref{propEstimateLichUpper}, we have $\phi' \in L^r(M,
	\bR_+)$.

As a consequence, we have defined a mapping $\Phi: L^r(M, \bR_+) \to L^r(M,
	\bR_+)$ given by the composition
\[
	\Phi(\phi) = (\mathrm{Lich} \circ \bL \circ \mathrm{Vect})(\phi).
\]
Our first step is to show that $\Phi$ maps a suitable set into itself.

Recall the constant $\delta$ from Proposition~\ref{propVolumeBound}:
\[
	\delta \definedas \mu_L - \frac{2}{\mu_V}\left(\frac{n-1}{n}\right)^2 \|d\tau\|_{L^n}^2 V_{\max}^{2/n} > 0.
\]
And define
\begin{equation}\label{eqDefR}
	R \definedas R(\sigma) > 0 \quad\text{by}\quad
	R^{2\frac{n-1}{n}} = \frac{1}{\delta} \|\sigma\|_{L^2}^2.
\end{equation}
This is motivated by Proposition~\ref{propVolumeBound}: $R$ is precisely the
value at which $\|\phi^N\|_{L^{\frac{N}{2}+1}}$ saturates
inequality~\eqref{eqVolumeBound}. We set
\[
	\Omega_0 \definedas \left\{\phi \in L^r(M, \bR_+) \mid
	\|\phi^N\|_{L^{\frac{N}{2}+1}} \leq R\right\}.
\]

\begin{proposition}\label{propStability}
	Let $q = \frac{2n(n-1)}{3n-2}$ be such that $\frac{1}{q} = \frac{1}{n} +
		\frac{1}{2(n-1)}$. If
	\[
		R \leq \left(\frac{\|d\tau\|_{L^n}}{\|d\tau\|_{L^q}}\right)^n V_{\max},
	\]
	then the set $\Omega_0$ is stable under $\Phi$.
\end{proposition}

\begin{proof}
	Let $\phi \in \Omega_0$. We set $W = \mathrm{Vect}(\phi)$ and $\phi' =
		\mathrm{Lich}(W)$. We first estimate the $L^2$-norm of $\bL W$ as
	follows. From the definition of the constant $\mu_V$, we have
	\[
		\mu_V \|W\|_{L^N}^2 \leq \frac{1}{2} \int_M |\bL W|^2 d\mu^g.
	\]

	We multiply the vector equation by $W$ and integrate over $M$ to get
	\begin{align*}
		\frac{1}{2} \int_M |\bL W|^2 d\mu^g
		 & = - \frac{n-1}{n} \int_M \phi^N \< d\tau, W\> d\mu^g                                                                                              \\
		 & \leq \frac{n-1}{n} \left\|\phi^N\right\|_{L^{\frac{N}{2}+1}} \|W\|_{L^N} \|d\tau\|_{L^q}                                                          \\
		 & \leq \frac{n-1}{n} \left\|\phi^N\right\|_{L^{\frac{N}{2}+1}} \left(\frac{1}{2\mu_V} \int_M |\bL W|^2 d\mu^g\right)^{\frac{1}{2}} \|d\tau\|_{L^q},
	\end{align*}
	where we have used H\"older's inequality to pass from the first line to the
	second one. As a consequence, we conclude that
	\[
		\frac{\mu_V}{2} \int_M |\bL W|^2 d\mu^g \leq \left(\frac{n-1}{n}\right)^2 \|d\tau\|_{L^q}^2 R^2,
	\]
	where we have estimated $\|\phi^N\|_{L^{\frac{N}{2}+1}}$ from above by $R$.
	Next, from the estimate~\eqref{eqDefMuL}, we have
	\begin{align*}
		\mu_L \left\|(\phi')^N\right\|_{L^{\frac{N}{2}+1}}^{2 \frac{n-1}{n}}
		 & \leq \int_M |\sigma + \bL W|^2 d\mu^g                                                               \\
		 & \leq \int_M |\sigma|^2 d\mu^g + \int_M |\bL W|^2 d\mu^g                                             \\
		 & \leq \int_M |\sigma|^2 d\mu^g + \frac{2}{\mu_V} \left(\frac{n-1}{n}\right)^2 \|d\tau\|_{L^q}^2 R^2.
	\end{align*}
	So $\phi' \in \Omega_0$ as soon as
	\[
		\|\sigma\|_{L^2}^2 + \frac{2}{\mu_V} \left(\frac{n-1}{n}\right)^2 \|d\tau\|_{L^q}^2 R^2 \leq \mu_L R^{2 \frac{n-1}{n}}.
	\]
	Substituting the definition~\eqref{eqDefR} of $R$, this condition becomes
	\[
		\frac{2}{\mu_V} \left(\frac{n-1}{n}\right)^2 \|d\tau\|_{L^q}^2 R^{\frac{2}{n}}
		\leq \mu_L - \delta
		= \frac{2}{\mu_V}\left(\frac{n-1}{n}\right)^2 \|d\tau\|_{L^n}^2 V_{\max}^{2/n},
	\]
	which is exactly the assumption $R \leq (\|d\tau\|_{L^n}/\|d\tau\|_{L^q})^n
		V_{\max}$.
\end{proof}

We next prove that, after repeated applications of $\Phi$, the set $\Omega_0$
is mapped into a bounded subset of $L^r(M, \bR)$.

\begin{proposition}\label{propBootstrap}
	Assume that $\|\sigma\|_{L^{2p}} \leq 1$. There exist an integer $K \geq 0$
	depending only on $n$, and a constant $C > 0$ depending on $(M, g)$, $p$,
	$V_{\max}$, $\|d\tau\|_{L^n}$, and $\|\sigma\|_{L^{2p}}$, such that for all
	$\phi \in \Phi^K\left(\Omega_0\right)$, where $\Phi^K$ is the $K$-th iterate of $\Phi$,
	\[
		\left\|\phi^N\right\|_{L^{2p}} \leq C \|\sigma\|_{L^{2p}}^{\frac{n}{n-1}}.
	\]
	In particular, there exists a constant $C' > 0$ such that
	\[
		\|\bL W\|_{L^{2p}} \leq C' \|\sigma\|_{L^{2p}}^{\frac{n}{n-1}}
	\]
	for all $W = \mathrm{Vect}(\phi)$, with $\phi \in \Phi^K(\Omega_0)$.
\end{proposition}

\begin{proof}
	By Proposition~\ref{propStability}, $\Phi(\Omega_0) \subset \Omega_0$. We
	construct a decreasing sequence of closed subsets $\Omega_k \subset \Omega_0$
	and an increasing sequence of exponents $p_0 < p_1 < \cdots$ defined by
	\[
		p_0 \definedas \frac{N}{2}+1,
		\qquad
		\frac{1}{p_{k+1}} - 1 = \frac{n}{n-1} \left(\frac{1}{p_k} - 1\right),
	\]
	so that
	\begin{equation}\label{eqBootstrap}
		\frac{1}{p_k} = 1 - \left(\frac{n}{n-1}\right)^k \left(1 - \frac{1}{p_0}\right)
		= 1 - \left(\frac{n}{n-1}\right)^k \frac{n}{2(n-1)}.
	\end{equation}

	From formula~\eqref{eqBootstrap}, $1/p_k \to -\infty$, so there exists a
	smallest integer $K_0 \geq 0$ such that $\frac{1}{p_{K_0}} \leq \frac{1}{n}$.
	We claim that $\frac{1}{p_{K_0}} \in \left(0, \frac{1}{n}\right)$. Indeed, for
	the positivity, the recurrence gives, for any $k < K_0$ (i.e.\
	$\frac{1}{p_k} > \frac{1}{n}$),
	\[
		\frac{1}{p_{k+1}} = 1 + \frac{n}{n-1}\left(\frac{1}{p_k} - 1\right)
		> 1 + \frac{n}{n-1}\left(\frac{1}{n} - 1\right) = 0,
	\]
	so in particular $\frac{1}{p_{K_0}} > 0$. For the strict upper bound, note that
	$\frac{1}{p_{K_0}} \in \left\{0, \frac{1}{n}\right\}$ would require, by
	formula~\eqref{eqBootstrap}, that $\left(\frac{n}{n-1}\right)^{K_0+1} = 2$ or
	$\left(\frac{n}{n-1}\right)^{K_0+2} = 2$. Both are impossible since
	$n/(n-1)$ is rational whereas $2^{1/m}$ is irrational for every $m \geq 2$.
	Hence $\frac{1}{p_{K_0}} \in \left(0, \frac{1}{n}\right)$, i.e.\ $p_{K_0} > n$.

	We define $\Omega_k$ by associating to each exponent $p_k$ a radius $R_k
		\definedas C_k \|\sigma\|_{L^{2p}}^{\frac{n}{n-1}}$, where $C_k$ is a positive
	constant to be specified inductively:
	\[
		\Omega_k \definedas \left\{\phi \in L^r(M, \bR_+) \;\middle|\;
		\|\phi^N\|_{L^{p_j}} \leq R_j\ \text{for}\ 0 \leq j \leq k\right\}.
	\]
	For the base case $k = 0$, we set $R_0$ by
	\[
		R_0^{2\frac{n-1}{n}} \definedas \frac{1}{\delta} \|\sigma\|_{L^{2p}}^2 \vol(M, g)^{1-\frac{1}{p}}
		\geq \frac{1}{\delta} \|\sigma\|_{L^2}^2,
	\]
	where the last inequality is H\"older's. By definition of $\Omega_0$, any $\phi
		\in \Omega_0$ satisfies $\|\phi^N\|_{L^{p_0}} \leq R_0$.

	\medskip\noindent\textit{Inductive step ($p_k < n$).}
	Let $\phi \in \Omega_k$ and $W = \mathrm{Vect}(\phi)$. Setting $q_k \definedas
		\left(\frac{1}{p_k}+\frac{1}{n}\right)^{-1}$ and applying elliptic regularity
	(Proposition~\ref{propVector}) followed by the Sobolev embedding theorem gives
	\begin{equation}\label{eqEstimateVector}
		\|\bL W\|_{L^{p_k}} \lesssim \|\phi^N\|_{L^{p_k}} \|d\tau\|_{L^n} \leq R_k \|d\tau\|_{L^n}.
	\end{equation}
	Since $p_k < n$, Proposition~\ref{propEstimateLichUpper} applied to $\phi' =
		\Phi(\phi)$ gives
	\[
		\|(\phi')^N\|_{L^{p_{k+1}}}^{2\frac{n-1}{n}}
		\lesssim \|A\|_{L^{p_k}}^2
		\lesssim \|\sigma\|_{L^{p_k}}^2 + \|\bL W\|_{L^{p_k}}^2
		\lesssim \|\sigma\|_{L^{2p}}^2 + R_k^2,
	\]
	where the exponents $p_k$ and $p_{k+1}$ satisfy $\frac{1}{p_k} = \frac{1}{n} +
		\frac{n-1}{n} \frac{1}{p_{k+1}}$, which is equivalent to the recurrence
	defining $(p_k)$. Let $\Lambda_{k+1}^2$ denote the implicit constant:
	\[
		\|(\phi')^N\|_{L^{p_{k+1}}}^{2\frac{n-1}{n}}
		\leq \Lambda_{k+1}^2 \left(\|\sigma\|_{L^{2p}}^2 + R_k^2\right).
	\]
	Since $R_k = C_k \|\sigma\|_{L^{2p}}^{\frac{n}{n-1}}$ and $\|\sigma\|_{L^{2p}}
		\leq 1$, we get
	\begin{align*}
		\|(\phi')^N\|_{L^{p_{k+1}}}^{2\frac{n-1}{n}}
		 & \leq \Lambda_{k+1}^2 \|\sigma\|_{L^{2p}}^2 \left(1 + C_k^2 \|\sigma\|_{L^{2p}}^{\frac{2}{n-1}}\right)
		\leq \Lambda_{k+1}^2 \left(1 + C_k^2\right) \|\sigma\|_{L^{2p}}^2.
	\end{align*}
	Setting $C_{k+1}$ by $C_{k+1}^{2\frac{n}{n-1}} \definedas \Lambda_{k+1}^2 (1 + C_k^2)$, we
	obtain $\|(\phi')^N\|_{L^{p_{k+1}}} \leq R_{k+1}$, hence $\Phi(\Omega_k) \subset
		\Omega_{k+1}$ and by induction $\Phi^k(\Omega_0) \subset \Omega_k$ for all
	$0 \leq k \leq K_0$.

	\medskip\noindent\textit{Terminal step ($p_{K_0} \geq n$).}
	We apply Proposition~\ref{propEstimateLichUpper} (second case) to
	$\phi \in \Omega_{K_0}$: since $p_{K_0} \geq n$, for any $r \geq 1$ we have
	\[
		\|(\Phi(\phi))^N\|_{L^r}^{2\frac{n-1}{n}}
		\lesssim \|A\|_{L^{p_{K_0}}}^2
		\lesssim \|\sigma\|_{L^{2p}}^2 + R_{K_0}^2
		\leq (1 + C_{K_0}^2)\|\sigma\|_{L^{2p}}^2,
	\]
	where the last inequality uses $\|\sigma\|_{L^{2p}} \leq 1$. Since $2(n-1)/n =
		1 + 2/N$, we conclude that
	\[
		\|(\Phi(\phi))^N\|_{L^r} \leq C_r \|\sigma\|_{L^{2p}}^{\frac{n}{n-1}},
	\]
	where $C_r$ depends on $(M, g)$, $p$, $K_0$, and $r$, but not on $\phi$.
	Setting $K \definedas K_0 + 1$, we have shown that every $\phi \in
		\Phi^K(\Omega_0) \subset \Phi(\Omega_{K_0})$ satisfies the claimed bound on
	$\|\phi^N\|_{L^r}$. Taking $r = 2p$ gives the first part of the proposition.

	\medskip

	For the second part, let $\phi \in \Phi^K(\Omega_0)$ and $W =
		\mathrm{Vect}(\phi)$. Applying the estimate~\eqref{eqEstimateVector} with $p_k$
	replaced by $2p$ and using the bound just established yields
	\[
		\|\bL W\|_{L^{2p}} \lesssim \|\phi^N\|_{L^{2p}} \|d\tau\|_{L^n}
		\lesssim \|\sigma\|_{L^{2p}}^{\frac{n}{n-1}} \|d\tau\|_{L^n},
	\]
	which completes the proof.
\end{proof}

In the following two lemmas, we fix $\phi_1, \phi_2 \in \Phi^K(\Omega_0)$ and
adopt the notation
\[
	W_i \definedas \mathrm{Vect}(\phi_i),
	\quad
	\phi'_i \definedas \Phi(\phi_i) = \mathrm{Lich}(W_i),
	\quad
	A_i \definedas |\sigma + \bL W_i|,
	\quad i = 1, 2,
\]
as well as $\psi \definedas \phi_1 - \phi_2$ and $\psi' \definedas \phi'_1 -
	\phi'_2$.

\begin{lemma}\label{lmEstimate_LW1-LW2}
	We have
	\[
		\left\|\bL W_1 - \bL W_2\right\|_{L^{2p}}
		\lesssim \|\psi\|_{L^r} \max \left\{\|\phi_1\|_{L^r}^{N-1}, \|\phi_2\|_{L^r}^{N-1}\right\} \|d\tau\|_{L^n}.
	\]
	In particular, there exists a constant $\Lambda_V > 0$ such that
	\[
		\left\|\bL W_1 - \bL W_2\right\|_{L^{2p}}
		\leq \Lambda_V \|\psi\|_{L^r} \|\sigma\|_{L^{2p}}^{\frac{n+2}{2(n-1)}}.
	\]
\end{lemma}

\begin{proof}
	Subtracting the equations satisfied by $W_1$ and $W_2$, we have that $W_1-W_2$
	satisfies:
	\[
		\Delta_{\bL} (W_1-W_2) = \frac{n-1}n (\phi_1^N - \phi_2^N)d\tau.
	\]
	Proceeding as in the proof of Proposition~\ref{propBootstrap}, we have
	\[
		\left\|\bL W_1 - \bL W_2\right\|_{L^{2p}}
		\lesssim \left\|\phi_1^N - \phi_2^N\right\|_{L^{2p}} \|d\tau\|_{L^n}.
	\]
	And, by calcuations similar to the ones in
	Proposition~\ref{propLowerBoundLich},
	\[
		\phi_1^N - \phi_2^N = N (\phi_1 - \phi_2) \int_0^1 \left(\lambda \phi_1 + (1-\lambda) \phi_2\right)^{N-1} d\lambda.
	\]
	By H\"older's inequality together with the convexity of the norm, we have
	\begin{align*}
		\left\|\phi_1^N - \phi_2^N\right\|_{L^{2p}}
		 & \leq N \|\phi_1 - \phi_2\|_{L^{2Np}} \int_0^1 \left\|\left(\lambda \phi_1 + (1-\lambda) \phi_2\right)\right\|_{L^{2Np}}^{N-1} d\lambda \\
		 & \leq N \|\psi\|_{L^r} \max \left\{\|\phi_1\|_{L^r}^{N-1}, \|\phi_2\|_{L^r}^{N-1}\right\}.
	\end{align*}
	This gives the first estimate. For the second, Proposition~\ref{propBootstrap}
	gives
	\[
		\|\phi_i\|_{L^r}^N \lesssim \|\sigma\|_{L^{2p}}^{\frac{n}{n-1}},
		\quad\text{hence}\quad
		\|\phi_i\|_{L^r}^{N-1} \lesssim \|\sigma\|_{L^{2p}}^{\frac{n}{n-1} \cdot \frac{N-1}{N}}
		= \|\sigma\|_{L^{2p}}^{\frac{n+2}{2(n-1)}},
	\]
	and the second estimate follows.
\end{proof}

\begin{lemma}\label{lmMajorationPsi}
	There exists a constant $\Lambda_L = \Lambda_L(M, g, p) > 0$ such that
	\[
		\|\psi'\|_{L^{2Np}}
		\leq \Lambda_L \frac{\left\|\bL W_1 - \bL W_2\right\|_{L^{2p}} \left\|2 \sigma + \bL W_1 + \bL W_2\right\|_{L^{2p}}}{\left(\inf_M \phi'_1\right)^{N+1}}.
	\]
\end{lemma}

\begin{proof}
	We assume here for simplicity that $\scal \geq \mu > 0$. The general case can
	be handled by means similar to those in
	Proposition~\ref{propEstimateLichUpper}.

	We subtract the equations satisfied by $\phi'_1$ and $\phi'_2$:
	\[
		\left\lbrace
		\begin{aligned}
			- \frac{4(n-1)}{n-2} \Delta \phi'_1 + \scal \phi'_1 + \frac{n-1}{n} \tau^2 (\phi'_1)^{N-1} & = \frac{A_1^2}{(\phi'_1)^{N+1}}, \\
			- \frac{4(n-1)}{n-2} \Delta \phi'_2 + \scal \phi'_2 + \frac{n-1}{n} \tau^2 (\phi'_2)^{N-1} & = \frac{A_2^2}{(\phi'_2)^{N+1}},
		\end{aligned}
		\right.
	\]
	and get
	\begin{align*}
		 & - \frac{4(n-1)}{n-2} \Delta \psi' + \scal \psi' + \frac{n-1}{n} \tau^2 \left((\phi'_1)^{N-1} - (\phi'_2)^{N-1}\right)                  \\
		 & \qquad\qquad\qquad = \frac{A_1^2}{(\phi'_1)^{N+1}} - \frac{A_2^2}{(\phi'_2)^{N+1}}                                                     \\
		 & \qquad\qquad\qquad = \frac{A_1^2 - A_2^2}{(\phi'_1)^{N+1}} + A_2^2 \left(\frac{1}{(\phi'_1)^{N+1}} - \frac{1}{(\phi'_2)^{N+1}}\right).
	\end{align*}

	Let $\alpha \geq 1$ be a constant to be chosen later. We multiply the previous
	equation by $(\psi')^{2\alpha-1} \definedas |\psi'|^{2\alpha-2} \psi'$ and
	integrate over $M$. We get
	\begin{align*}
		 & \int_M (\psi')^{2\alpha-1}\left(- \frac{4(n-1)}{n-2} \Delta \psi' + \scal \psi'\right)\, d\mu^g                                   \\
		 & \qquad\qquad = -\frac{n-1}{n} \int_M \tau^2 (\psi')^{2\alpha-1} \left((\phi'_1)^{N-1} - (\phi'_2)^{N-1}\right) \, d\mu^g          \\
		 & \qquad\qquad\qquad + \int_M A_2^2 (\psi')^{2\alpha-1} \left(\frac{1}{(\phi'_1)^{N+1}} - \frac{1}{(\phi'_2)^{N+1}}\right)\, d\mu^g \\
		 & \qquad\qquad\qquad + \int_M (\psi')^{2\alpha-1} \frac{A_1^2 - A_2^2}{(\phi'_1)^{N+1}}\, d\mu^g.
	\end{align*}
	As $(\psi')^{2\alpha-1}$ has the same sign as $\psi' = \phi'_1 - \phi'_2$, the
	first two terms on the right-hand side are non-positive. So,
	\[
		\int_M (\psi')^{2\alpha-1}\left(- \frac{4(n-1)}{n-2} \Delta \psi' + \scal \psi'\right)\, d\mu^g
		\leq \int_M (\psi')^{2\alpha-1} \frac{A_1^2 - A_2^2}{(\phi'_1)^{N+1}}\, d\mu^g.
	\]
	Note also that
	\begin{align*}
		\int_M (\psi')^{2\alpha-1} (-\Delta \psi') \, d\mu^g
		 & = \int_M \left\<d (\psi')^{2\alpha-1}, d\psi'\right\>\, d\mu^g                \\
		 & = (2\alpha - 1) \int_M |\psi'|^{2\alpha-2}\left| d\psi'\right|^2\, d\mu^g     \\
		 & = \frac{2\alpha-1}{\alpha^2} \int_M \left| d|\psi'|^\alpha\right|^2\, d\mu^g.
	\end{align*}
	All in all, we have obtained
	\[
		\int_M \left(\frac{4(n-1)}{n-2} \frac{2\alpha-1}{\alpha^2} \left| d|\psi'|^\alpha\right|^2 + \scal |\psi'|^{2\alpha}\right)\, d\mu^g
		\leq \int_M (\psi')^{2\alpha-1} \frac{A_1^2 - A_2^2}{(\phi'_1)^{N+1}}\, d\mu^g.
	\]
	As the scalar curvature of $g$ is bounded from below, we can use the Sobolev
	embedding theorem to get
	\[
		\|\psi'\|_{L^{N\alpha}}^{2\alpha} = \left\||\psi'|^\alpha\right\|_{L^N}^2
		\lesssim \int_M (\psi')^{2\alpha-1} \frac{A_1^2 - A_2^2}{(\phi'_1)^{N+1}}\, d\mu^g.
	\]
	We apply H\"older's inequality to the right-hand side with exponents $a, b, c,
		d \in (1, \infty)$ satisfying $\frac{1}{a} + \frac{1}{b} + \frac{1}{c} +
		\frac{1}{d} = 1$. We choose
	\[
		(2\alpha-1)a = N\alpha,
		\quad
		b = c = 2p,
		\quad
		\alpha = 2p,
	\]
	so that $\frac{1}{a} = \frac{2\alpha - 1}{N\alpha}$ and
	\[
		\frac{1}{d} = 1 - \frac{2\alpha-1}{N\alpha} - \frac{1}{p}
		= \frac{2}{n} - \left(1 - \frac{1}{2N}\right)\frac{1}{p}
		\in \left(0, \frac{2}{n}\right),
	\]
	where the bounds follow from $p > n/2$. Using $A_1^2 - A_2^2 = \< \bL W_1 - \bL
		W_2, 2\sigma + \bL W_1 + \bL W_2\>$, we obtain
	\[
		\|\psi'\|_{L^{2Np}}^{2\alpha}
		\lesssim \frac{1}{\left(\inf_M \phi'_1\right)^{N+1}}
		\left\|\psi'\right\|_{L^{2Np}}^{2\alpha-1}
		\left\|\bL W_1 - \bL W_2\right\|_{L^{2p}}
		\left\|2 \sigma + \bL W_1 + \bL W_2\right\|_{L^{2p}},
	\]
	which gives the claimed estimate after dividing both sides by
	$\|\psi'\|_{L^{2Np}}^{2\alpha-1}$.
\end{proof}

We can now finish the proof of the main result of the paper.

\begin{proof}[Proof of Theorem~\ref{thmMain}]
	Combining Lemmas~\ref{lmEstimate_LW1-LW2} and~\ref{lmMajorationPsi} and
	assuming $\|\sigma\|_{L^{2p}} \leq 1$, we obtain
	\[
		\|\psi'\|_{L^r}
		\leq \frac{\Lambda_L \Lambda_V}{\left(\inf_M \phi'_1\right)^{N+1}}
		\|\psi\|_{L^r} \|\sigma\|_{L^{2p}}^{\frac{n+2}{2(n-1)}}
		\left\|2 \sigma + \bL W_1 + \bL W_2\right\|_{L^{2p}}.
	\]
	Define the contraction constant
	\[
		\lambda \definedas \frac{\Lambda_L \Lambda_V}{\left(\inf_M \phi'_1\right)^{N+1}}
		\left\|2 \sigma + \bL W_1 + \bL W_2\right\|_{L^{2p}}
		\|\sigma\|_{L^{2p}}^{\frac{n+2}{2(n-1)}}.
	\]
	We show that $\lambda < 1$ when $\|\sigma\|_{L^{2p}}$ is sufficiently small.

	\medskip\noindent\textit{Bounding $\|2\sigma + \bL W_1 + \bL W_2\|_{L^{2p}}$.}
	From Proposition~\ref{propBootstrap}, $\|\bL W_i\|_{L^{2p}} \leq C'
		\|\sigma\|_{L^{2p}}^{\frac{n}{n-1}}$ for $i = 1, 2$. Since $\|\sigma\|_{L^{2p}}
		\leq 1$ implies $\|\sigma\|_{L^{2p}}^{\frac{n}{n-1}} \leq \|\sigma\|_{L^{2p}}$, we get
	\[
		\left\|2 \sigma + \bL W_1 + \bL W_2\right\|_{L^{2p}}
		\leq 2 \|\sigma\|_{L^{2p}} + 2 C' \|\sigma\|_{L^{2p}}^{\frac{n}{n-1}}
		\leq (2 + 2 C') \|\sigma\|_{L^{2p}},
	\]
	so
	\[
		\lambda \leq \frac{2\Lambda_L \Lambda_V (1 +C')}{\left(\inf_M \phi'_1\right)^{N+1}}
		\|\sigma\|_{L^{2p}}^{\frac{3n}{2(n-1)}}.
	\]

	\medskip\noindent\textit{Bounding $\inf_M \phi'_1$ from below.}
	Since $A_1 = \sigma + \bL W_1$, Proposition~\ref{propBootstrap} gives
	\[
		\|A_1\|_{L^{2p}} \leq \|\sigma\|_{L^{2p}} + C' \|\sigma\|_{L^{2p}}^{\frac{n}{n-1}},
	\]
	so, if $\|\sigma\|_{L^{2p}}$ is small enough, the condition $\|A_1\|_{L^{2p}}
		\leq C_{g,\tau}^{-1/2}$ in Lemma~\ref{lmSubsolution} is fulfilled. Moreover,
	the hypothesis $\|\sigma\|_{L^{2p}} \leq \omega_0 \|\sigma\|_{L^2}$ and $\|\bL
		W_1\|_{L^{2p}} \leq C' \|\sigma\|_{L^{2p}}^{\frac{n}{n-1}}$ give
	\[
		\omega_1 \definedas \frac{\|A_1\|_{L^{2p}}}{\|A_1\|_{L^2}}
		\leq \frac{\|\sigma\|_{L^{2p}} + C' \|\sigma\|_{L^{2p}}^{\frac{n}{n-1}}}{\|\sigma\|_{L^2}}
		\leq \omega_0 \left(1 + C' \|\sigma\|_{L^{2p}}^{\frac{1}{n-1}}\right).
	\]
	So $\omega_1$ is bounded above independently of $\sigma$ (for
	$\|\sigma\|_{L^{2p}} \leq 1$). Proposition~\ref{propLowerBoundLich} then yields
	\begin{align*}
		\inf_M \phi_1'
		 & \geq \mu_{g, \tau}\, \omega_1^{-\frac{3n-2}{2(n-1)}} \|A_1\|_{L^2}^{\frac{n-2}{2(n-1)}} \\
		 & \gtrsim \|\sigma\|_{L^2}^{\frac{n-2}{2(n-1)}}
		\gtrsim \|\sigma\|_{L^{2p}}^{\frac{n-2}{2(n-1)}}.
	\end{align*}

	\medskip\noindent\textit{Conclusion.}
	Combining the two bounds above,
	\[
		\lambda \lesssim
		\frac{\|\sigma\|_{L^{2p}}^{\frac{3n}{2(n-1)}}}
		{\left(\|\sigma\|_{L^{2p}}^{\frac{n-2}{2(n-1)}}\right)^{N+1}}
		= \|\sigma\|_{L^{2p}}^{\frac{1}{n-1}},
	\]
	which tends to zero as $\|\sigma\|_{L^{2p}} \to 0$. In particular, for
	$\|\sigma\|_{L^{2p}}$ small enough, $\lambda < 1$, and we have
	\[
		\left\|\Phi(\phi_1) - \Phi(\phi_2)\right\|_{L^r} \leq \lambda \|\phi_1 - \phi_2\|_{L^r}
		\quad\text{for all } \phi_1, \phi_2 \in \overline{\Phi^K(\Omega_0)}.
	\]
	The set $\Omega \definedas \overline{\Phi^K(\Omega_0)}$ is closed in $L^r(M,
		\bR)$ and $\Phi$-invariant. By the Banach fixed point theorem, $\Phi$ has a
	unique fixed point $\phi \in \Omega$. Setting $W \definedas
		\mathrm{Vect}(\phi)$, the pair $(\phi, W)$ is a solution to the conformal
	constraint equations~\eqref{eqConstraints}.

	Finally, any solution $(\phi, W)$ to~\eqref{eqConstraints} with $V(\phi, W)
		\leq V_{\max}$ belongs to $\Omega_0$ by Proposition~\ref{propVolumeBound}, and
	hence to $\Omega$. Since $\Phi$ has a unique fixed point in $\Omega$, this
	solution coincides with $(\phi, W)$, establishing uniqueness.
\end{proof}